\documentclass{llncs}
\usepackage[latin2]{inputenc}
\usepackage[english]{babel}
\usepackage{dsfont, complexity}
\usepackage{amsmath, lmodern}
\usepackage[title]{appendix}
\usepackage{amssymb}
\usepackage{lmodern}
\usepackage{float}
\usepackage{graphicx}
\usepackage{subfigure} 
\usepackage{tikz}
\usepackage{tkz-graph}
\usepackage{xcolor}
\usepackage{url}
\usepackage{paralist, enumerate}
\usepackage{refcount}
\usepackage[textsize=tiny,textwidth=2.2cm,shadow]{todonotes}

\usepackage{thmtools, mathtools}
\usepackage{thm-restate}

\usetikzlibrary{arrows,decorations.markings,patterns,calc, positioning, backgrounds}
\tikzstyle{vertex} = [circle, draw=black, scale=0.7]
\tikzstyle{edgelabel} = [circle, fill=white, scale=0.8]

\usepackage{amsthm}

\usepackage{algpseudocode}
\usepackage[ruled]{algorithm} 

\begin{document}
\title{An Integer Programming Approach to the Student-Project Allocation Problem with Preferences over Projects}

\author{
David Manlove\thanks{Supported by grant EP/P028306/1 from the Engineering and Physical Sciences Research Council. Orcid ID: 0000-0001-6754-7308.}
\and
Duncan Milne
\and Sofiat Olaosebikan\thanks{Supported by a College of Science and Engineering Scholarship, University of Glasgow. Orcid ID: 0000-0002-8003-7887.}
}

\institute{School of Computing Science, University of Glasgow, \\e-mail: \texttt{David.Manlove@glasgow.ac.uk, Duncan.Milne1@gmail.com, s.olaosebikan.1@research.gla.ac.uk}
}
\maketitle

\begin{abstract}
The Student-Project Allocation problem with preferences over Projects ({\scriptsize SPA-P}) involves sets of students, projects and lecturers, where the students and lecturers each have preferences over the projects. In this context, we typically seek a stable matching of students to projects (and lecturers). However, these stable matchings can have different sizes, and the problem of finding a maximum stable matching ({\scriptsize MAX-SPA-P}) is NP-hard. There are two known approximation algorithms for {\scriptsize MAX-SPA-P}, with performance guarantees of $2$ and $\frac{3}{2}$. In this paper, we describe an Integer Programming (IP) model to enable {\scriptsize MAX-SPA-P} to be solved optimally. Following this, we present results arising from an empirical analysis that investigates how the solution produced by the approximation algorithms compares to the optimal solution obtained from the IP model, with respect to the size of the stable matchings constructed, on instances that are both randomly-generated and derived from real datasets. Our main finding is that the $\frac{3}{2}$-approximation algorithm finds stable matchings that are very close to having maximum cardinality.
\end{abstract}

\thispagestyle{empty}
\setcounter{page}{1}
\pagestyle{headings}
\section{Introduction}
Matching problems, which generally involve the assignment of a set of agents to another set of agents based on preferences, have wide applications in many real-world settings. One such application can be seen in an educational context, e.g., the allocation of pupils to schools, school-leavers to universities and students to projects. In the context of allocating students to projects, university lecturers propose a range of projects, and each student is required to provide a preference over the available projects that she finds acceptable. Lecturers may also have preferences over the students that find their project acceptable and/or the projects that they offer. There may also be upper bounds on the number of students that can be assigned to a particular project, and the number of students that a given lecturer is willing to supervise. The problem then is to allocate students to projects based on these preferences and capacity constraints -- the so-called \textit{Student-Project Allocation problem} ({\scriptsize SPA}) \cite{Man13,RGSA17}. 

Two major models of {\scriptsize SPA} exist in the literature: one permits preferences only from the students \cite{AB03,HSVS05,KIMS15}, while the other permits preferences from the students and lecturers \cite{AIM07,Kaz02}. Given the large number of students that are typically involved in such an allocation process, many university departments seek to automate the allocation of students to projects. Examples include the School of Computing Science, University of Glasgow \cite{KIMS15}, the Faculty of Science, University of Southern Denmark \cite{CFG17}, the Department of Computing Science, University of York \cite{Kaz02}, and elsewhere \cite{AB03,RGSA17,HSVS05}. 

In general, we seek a \textit{matching}, which is a set of agent pairs who find one another acceptable that satisfies the capacities of the agents involved. For matching problems where preferences exist from the two sets of agents involved (e.g., junior doctors and hospitals in the classical \textit{Hospitals-Residents problem} ({\scriptsize HR}) \cite{GS62}, or students and lecturers in the context of {\scriptsize SPA}), it has been argued that the desired property for a matching one should seek is that of \textit{stability} \cite{Rot84}. Informally, a \textit{stable matching} ensures that no acceptable pair of agents who are not matched together would rather be assigned to each other than remain with their current assignees.

Abraham, Irving and Manlove \cite{AIM07} proposed two linear-time algorithms to find a stable matching in a variant of {\scriptsize SPA} where students have preferences over projects, whilst lecturers have preferences over students. The stable matching produced by the first algorithm is student-optimal (that is, students have the best possible projects that they could obtain in any stable matching) while the one produced by the second algorithm is lecturer-optimal (that is, lecturers have the best possible students that they could obtain in any stable matching).
 
Manlove and O'Malley \cite{MO08} proposed another variant of {\scriptsize SPA} where both students and lecturers have preferences over projects, referred to as {\scriptsize SPA-P}. In their paper, they formulated an appropriate stability definition for {\scriptsize SPA-P}, and they showed that stable matchings in this context can have different sizes. Moreover, in addition to stability, a very important requirement in practice is to match as many students to projects as possible. Consequently, Manlove and O'Malley \cite{MO08} proved that the problem of finding a maximum cardinality stable matching, denoted {\scriptsize MAX-SPA-P}, is NP-hard. Further, they gave a polynomial-time $2$-approximation algorithm for {\scriptsize MAX-SPA-P}. Subsequently, Iwama, Miyazaki and Yanagisawa \cite{IMY12} described an improved approximation algorithm with an upper bound of $\frac{3}{2}$, which builds on the one described in \cite{MO08}. In addition, Iwama \textit{et al.}~\cite{IMY12} showed that {\scriptsize MAX-SPA-P} is not approximable within $\frac{21}{19} - \epsilon$, for any $\epsilon > 0$, unless P = NP. For the upper bound, they modified Manlove and O'Malley's algorithm \cite{MO08} using Kir\'{a}ly's idea \cite{Kir11} for the approximation algorithm to find a maximum stable matching in a variant of the \textit{Stable Marriage problem}.

Considering the fact that the existing algorithms for {\scriptsize MAX-SPA-P} are only guaranteed to produce an approximate solution, we seek another technique to enable {\scriptsize MAX-SPA-P} to be solved optimally. Integer Programming (IP) is a powerful technique for producing optimal solutions to a range of NP-hard optimisation problems, with the aid of commercial optimisation solvers, e.g., Gurobi \cite{ZZZ20}, GLPK \cite{ZZZ21} and CPLEX \cite{ZZZ22}. These solvers can allow IP models to be solved in a reasonable amount of time, even with respect to problem instances that occur in practical applications.


\paragraph{\textbf{Our Contribution.}} In Sect.~\ref{sect:ip-model}, we describe an IP model to enable {\scriptsize MAX-SPA-P} to be solved optimally, and present a correctness result. In Sect.~\ref{sect:empirical-analysis}, we present results arising from an empirical analysis that investigates how the solution produced by the approximation algorithms compares to the optimal solution obtained from our IP model, with respect to the size of the stable matchings constructed, on instances that are both randomly-generated and derived from real datasets. These real datasets are based on actual student preference data and manufactured lecturer preference data from previous runs of student-project allocation processes at the School of Computing Science, University of Glasgow. We also present results showing the time taken by the IP model to solve the problem instances optimally. Our main finding is that the $\frac{3}{2}$-approximation algorithm finds stable matchings that are very close to having maximum cardinality. The next section gives a formal definition for {\scriptsize SPA-P}.

\section{Definitions and Preliminaries}
\label{sect:definitions}
We give a formal definition for {\scriptsize SPA-P} as described in the literature \cite{MO08}. An instance $\mathtt{I}$ of {\scriptsize SPA-P} involves a set $\mathcal{S} = \{s_1 , s_2, \ldots , s_{n_1}\}$ of \textit{students}, a set  $\mathcal{P} = \{p_1 , p_2, \ldots , p_{n_2}\}$ of \textit{projects} and a set $\mathcal{L} = \{l_1 , l_2, \ldots , l_{n_3}\}$ of \textit{lecturers}. Each lecturer $l_k \in \mathcal{L}$ offers a non-empty subset of projects, denoted by $P_k$. We assume that $P_1, P_2, \ldots, P_{n_3}$ partitions $\mathcal{P}$ (that is, each project is offered by one lecturer). Also, each student $s_i \in \mathcal{S}$ has an \textit{acceptable} set of projects $A_i \subseteq \mathcal{P}$.  We call a pair $(s_i, p_j) \in \mathcal{S} \times \mathcal{P}$ an \textit{acceptable pair} if $p_j \in A_i$. Moreover $s_i$ ranks $A_i$ in strict order of preference. Similarly, each lecturer $l_k$ ranks $P_k$ in strict order of preference. Finally, each project $p_j \in \mathcal{P}$ and lecturer $l_k \in \mathcal{L}$ has a positive capacity denoted by $c_j$ and $d_k$ respectively.  
%
%

An \textit{assignment} $M$ is a subset of $\mathcal{S} \times \mathcal{P}$ where $(s_i, p_j) \in M$ implies that $s_i$ finds $p_j$ acceptable (that is, $p_j \in A_i$). We define the \textit{size} of $M$ as the number of (student, project) pairs in $M$, denoted $|M|$. If $(s_i, p_j) \in M$, we say that $s_i$ is \textit{assigned to} $p_j$ and $p_j$ is \textit{assigned} $s_i$. Furthermore, we denote the project assigned to student $s_i$ in $M$ as $M(s_i)$ (if $s_i$ is unassigned in $M$ then $M(s_i)$ is undefined). Similarly, we denote the set of students assigned to project $p_j$ in $M$ as $M(p_j)$. For ease of exposition, if $s_i$ is assigned to a project $p_j$ offered by lecturer $l_k$, we may also say that $s_i$ is \textit{assigned to} $l_k$, and $l_k$ is \textit{assigned} $s_i$. Thus we denote the set of students assigned to $l_k$ in $M$ as $M(l_k)$. 

A project $p_j \in \mathcal{P}$ is \textit{full}, \textit{undersubscribed} or \textit{oversubscribed} in $M$ if $|M(p_j)|$ is equal to, less than or greater than $c_j$, respectively. The corresponding terms apply to each lecturer $l_k$ with respect to $d_k$.  We say that a project $p_j \in \mathcal{P}$ is \textit{non-empty} if $|M(p_j)| > 0$.

A \textit{matching} $M$ is an assignment such that $|M(s_i)| \leq 1$ for each $s_i \in \mathcal{S}$, $|M(p_j)| \leq c_j$ for each $p_j \in \mathcal{P}$, and $|M(l_k)| \leq d_k$ for each $l_k \in \mathcal{L}$ (that is, each student is assigned to at most one project, and no project or lecturer is oversubscribed). Given a matching $M$, an acceptable pair $(s_i, p_j) \in (\mathcal{S} \times \mathcal{P}) \setminus M$ is a \textit{blocking pair} of $M$ if the following conditions are satisfied:
\begin{enumerate}
\label{text:blocking-pair-definition}
\item[1.] either $s_i$ is unassigned in $M$ or $s_i$ prefers $p_j$ to $M(s_i)$, and $p_j$ is undersubscribed, and either
	\begin{enumerate}
	\item $s_i \in M(l_k)$ and $l_k$ prefers $p_j$ to $M(s_i)$, or
	\item $s_i \notin M(l_k)$ and $l_k$ is undersubscribed, or
	\item $s_i \notin M(l_k)$ and $l_k$ prefers $p_j$ to his worst non-empty project,	
	\end{enumerate}
	where $l_k$ is the lecturer who offers $p_j$.
\end{enumerate}

If such a pair were to occur, it would undermine the integrity of the matching as the student and lecturer involved would rather be assigned together than remain in their current assignment. With respect to the {\scriptsize SPA-P} instance given in Fig.~\ref{fig:spa-p-instance-1}, $M_1 = \{(s_1, p_3), (s_2, p_1)\}$ is clearly a matching. It is obvious that each of students $s_1$ and $s_2$ is matched to her first ranked project in $M_1$. Although $s_3$ is unassigned in $M_1$, the lecturer offering $p_3$ (the only project that $s_3$ finds acceptable) is assumed to be indifferent among those students who find $p_3$ acceptable. Also $p_3$ is full in $M_1$. Thus, we say that $M_1$ admits no blocking pair.
\begin{figure}[H]
\centering
\begin{tabular}{lll}
Student preferences & \quad \quad & Lecturer preferences \\ 
$s_1$: \quad  $p_3$ \; $p_2$ \;$p_1$ &  & $l_1$: \quad $p_2$ \; $p_1$ \\ 
 
$s_2$: \quad $p_1$ \; $p_2$ &  & $l_2$: \quad $p_3$ \\ 

$s_3$: \quad $p_3$ &  & 
\end{tabular}
\caption{An instance $\mathtt{I}_1$ of {\scriptsize SPA-P}. Each project has capacity $1$, whilst each of lecturer $l_1$ and $l_2$ has capacity $2$ and $1$ respectively.}
\label{fig:spa-p-instance-1}
\end{figure}
Another way in which a matching could be undermined is through a group of students acting together. Given a matching $M$, a \textit{coalition} is a set of students $\{s_{i_0}, \ldots, s_{i_{r-1}}\}$, for some $r \geq 2$ such that each student $s_{i_j}$ ($0 \leq j \leq r-1$) is assigned in $M$ and prefers $M(s_{i_{j+1}})$ to $M(s_{i_j})$, where addition is performed modulo $r$. With respect to Fig.~\ref{fig:spa-p-instance-1}, the matching $M_2 = \{(s_1, p_1), (s_2, p_2), (s_3, p_3)\}$ admits a coalition $\{s_1, s_2\}$, as students $s_1$ and $s_2$ would rather permute their assigned projects in $M_2$ so as to be better off. We note that the number of students assigned to each project and lecturer involved in any such swap remains the same after such a permutation. Moreover, the lecturers involved would have no incentive to prevent the switch from occurring since they are assumed to be indifferent between the students assigned to the projects they are offering. If a matching admits no coalition, we define such matching to be \textit{coalition-free}.

Given an instance $\mathtt{I}$ of {\scriptsize SPA-P}, we define a matching $M$ in $\mathtt{I}$ to be \textit{stable} if $M$ admits no blocking pair and is coalition-free. It turns out that with respect to this definition, stable matchings in $\mathtt{I}$ can have different sizes. Clearly, each of the matchings $M_1 = \{(s_1, p_3), (s_2, p_1)\}$ and $M_3 = \{(s_1, p_2), (s_2, p_1), (s_3, p_3)\}$ is stable in the {\scriptsize SPA-P} instance $\mathtt{I}_1$ shown in Fig.~\ref{fig:spa-p-instance-1}. The varying sizes of the stable matchings produced naturally leads to the problem of finding a maximum cardinality stable matching given an instance of {\scriptsize SPA-P}, which we denote by {\scriptsize MAX-SPA-P}. 
In the next section, we describe our IP model to enable {\scriptsize MAX-SPA-P} to be solved optimally.

\section{An IP model for {\footnotesize MAX-SPA-P}}
\label{sect:ip-model}
Let $\mathtt{I}$ be an instance of {\scriptsize SPA-P} involving a set $\mathcal{S} = \{s_1 , s_2, \ldots , s_{n_1}\}$ of students, a set  $\mathcal{P} = \{p_1 , p_2, \ldots , p_{n_2}\}$ of projects and a set $\mathcal{L} = \{l_1 , l_2, \ldots , l_{n_3}\}$ of lecturers. We construct an IP model $\mathtt{J}$ of $\mathtt{I}$ as follows. Firstly, we create binary variables $x_{i, j} \in \{0, 1\}$ $(1 \leq i \leq n_1, 1 \leq j \leq n_2)$ for each acceptable pair $(s_i, p_j) \in \mathcal{S} \times \mathcal{P}$ such that $x_{i, j}$ indicates whether $s_i$ is assigned to $p_j$ in a solution or not. Henceforth, we denote by $S$ a solution in the IP model $\mathtt{J}$, and we denote by $M$ the matching derived from $S$. If $x_{i,j} = 1$ under $S$ then intuitively $s_i$ is assigned to $p_j$ in $M$, otherwise $s_i$ is not assigned to $p_j$ in $M$. In what follows, we give the constraints to ensure that the assignment obtained from a feasible solution in $\mathtt{J}$ is a matching.
\paragraph{\textbf{Matching Constraints}.}
The feasibility of a matching can be ensured with the following three sets of constraints.
\begin{align}
\label{ineq:studentassignment}
\sum\limits_{p_{j} \in A_{i}} x_{i,j}  \leq 1 &\qquad (1 \leq i \leq n_1), \\
\label{ineq:projectcapacity}
\sum\limits_{i = 1}^{n_1} x_{i,j}   \leq c_j   & \qquad (1 \leq j \leq n_2), \\
\label{ineq:lecturercapacity}
\sum\limits_{i = 1}^{n_1} \; \sum\limits_{p_{j} \in P_k} x_{i,j}  \leq d_k  & \qquad (1 \leq k \leq  n_3)\enspace.
\end{align}
Note that \eqref{ineq:studentassignment} implies that each student $s_i \in \mathcal{S}$ is not assigned to more than one project, while \eqref{ineq:projectcapacity} and \eqref{ineq:lecturercapacity} implies that the capacity of each project $p_j \in \mathcal{P}$ and each lecturer $l_k \in \mathcal{L}$ is not exceeded.

We define $rank(s_i, p_j)$, the \textit{rank} of $p_j$ on $s_i$'s preference list, to be $r+1$ where $r$ is the number of projects that $s_i$ prefers to $p_j$. An analogous definition holds for $rank(l_k, p_j)$, the \textit{rank} of $p_j$ on $l_k$'s preference list. With respect to an acceptable pair $(s_i, p_j)$, we define $S_{i,j} = \{p_{j'} \in A_i: rank(s_i, p_{j'}) \leq rank(s_i, p_j)\}$, the set of projects that $s_i$ likes as much as $p_j$. 
For a project $p_j$ offered by lecturer $l_k \in \mathcal{L}$, we also define ${T_{k,j} = \{p_q \in P_k: rank(l_k, p_j) < rank(l_k, p_q)\}}$, the set of projects that are worse than $p_j$ on $l_k$'s preference list.
 
In what follows, we fix an arbitrary acceptable pair $(s_i, p_j)$ and we impose constraints to ensure that $(s_i, p_j)$ is not a blocking pair of the matching $M$ (that is, $(s_i, p_j)$ is not a type 1(a), type 1(b) or type 1(c) blocking pair of $M$). Firstly, let $l_k$ be the lecturer who offers $p_j$. 
\paragraph{\textbf{Blocking Pair Constraints.}} We define $\theta_{i,j} = 1 - \sum_{p_{j'} \in S_{i, j}}x_{i,j'}$.
Intuitively, $\theta_{i,j} = 1$ if and only if $s_i$ is unassigned in $M$ or prefers $p_j$ to $M(s_i)$. Next we create a binary variable $\alpha_j$ in $\mathtt{J}$ such that $\alpha_j = 1$ corresponds to the case when $p_j$ is undersubscribed in $M$. We enforce this condition by imposing the following constraint.
\begin{eqnarray}
\label{ineq:project-undersubscribed}
c_j \alpha_j \geq c_j - \sum\limits_{i' = 1}^{n_1} x_{i',j}\enspace,
\end{eqnarray}
where $\sum_{i' = 1}^{n_1} x_{i',j} = |M(p_j)|$. If $p_j$ is undersubscribed in $M$ then the RHS of \eqref{ineq:project-undersubscribed} is at least $1$, and this implies that $\alpha_j = 1$. Otherwise, $\alpha_j$ is not constrained. Now let $\gamma_{i,j,k} = \sum_{p_{j'} \in T_{k,j}} x_{i,j'}$.
Intuitively, if $\gamma_{i,j,k} = 1$ in $S$ then $s_i$ is assigned to a project $p_{j'}$ offered by $l_k$ in $M$, where $l_k$ prefers $p_j$ to $p_{j'}$. The following constraint ensures that $(s_i, p_j)$ does not form a type 1(a) blocking pair of $M$.
\begin{align}
\label{ineq:type-a-blockingpair}
\Aboxed{
\theta_{i,j} + \alpha_{j} + \gamma_{i,j,k} \leq 2\enspace.}
\end{align}
Note that if the sum of the binary variables in the LHS of \eqref{ineq:type-a-blockingpair} is less than or equal to $2$, this implies that at least one of the variables, say $\gamma_{i,j,k}$, is $0$. Thus the pair $(s_i, p_j)$ is not a type 1(a) blocking pair of $M$. 

Next we define $\beta_{i,k} = \sum_{p_{j'} \in P_k} x_{i,j'}.$
Clearly, $s_i$ is assigned to a project offered by $l_k$ in $M$ if and only if $\beta_{i,k} = 1$ in $S$.
Now we create a binary variable $\delta_k$ in $\mathtt{J}$ such that $\delta_k = 1$ in $S$ corresponds to the case when $l_k$ is undersubscribed in $M$. We enforce this condition by imposing the following constraint.
\begin{eqnarray}
\label{ineq:lecturer-undersubscribed}
d_k\delta_k \geq d_k - \sum\limits_{i' = 1}^{n_1} \; \sum\limits_{p_{j'} \in P_k} x_{i',j'}\enspace,
\end{eqnarray}
where $\sum_{i' = 1}^{n_1} \; \sum_{p_{j'} \in P_k} x_{i',j'} = |M(l_k)|$. If $l_k$ is undersubscribed in $M$ then the RHS of \eqref{ineq:lecturer-undersubscribed} is at least $1$, and this implies that $\delta_k = 1$. Otherwise, $\delta_k$ is not constrained. The following constraint ensures that $(s_i, p_j)$ does not form a type 1(b) blocking pair of $M$.
\begin{align}
\label{ineq:type-b-blockingpair}
\Aboxed{
\theta_{i,j} + \alpha_{j} + (1 - \beta_{i,k}) + \delta_k \leq 3\enspace.}
\end{align}

We define $D_{k,j} = \{p_{j'} \in P_k: rank(l_k, p_{j'}) \leq rank(l_k, p_j)\}$, the set of projects that $l_k$ likes as much as $p_j$.
Next, we create a binary variable $\eta_{j,k}$ in $\mathtt{J}$ such that $\eta_{j,k} = 1$ if $l_k$ is full and prefers $p_j$ to his worst non-empty project in $S$. We enforce this by imposing the following constraint.
\begin{eqnarray}
\label{ineq:lecturerfullconstraint}
d_k\eta_{j,k} \geq d_k - \sum\limits_{i' = 1}^{n_1} \; \sum\limits_{p_{j'} \in D_{k,j}} x_{i',j'}\enspace.
\end{eqnarray}
Finally, to avoid a type 1(c) blocking pair, we impose the following constraint.
\begin{align}
\label{ineq:type-c-blockingpair}
\Aboxed{
\theta_{i,j} + \alpha_{j} + (1 - \beta_{i,k}) + \eta_{j,k} \leq 3\enspace.}
\end{align}

Next, we give the constraints to ensure that the matching obtained from a feasible solution in $\mathtt{J}$ is coalition-free.
\paragraph{\textbf{Coalition Constraints.}} First, we introduce some additional notation. Given an instance $\mathtt{I}'$ of {\scriptsize{SPA-P}} and a matching $M'$ in $\mathtt{I}'$, we define the \textit{envy graph} $G(M') = (\mathcal{S}, A)$, where the vertex set $\mathcal{S}$ is the set of students in $\mathtt{I}'$, and the arc set $A = \{(s_i, s_{i'}): s_i \mbox{ prefers } M'(s_{i'}) \mbox{ to } M'(s_i)\}$.
It is clear that the matching $M_2 = \{(s_1, p_1), (s_2, p_2), (s_3, p_3)\}$ admits a coalition $\{s_1, s_2\}$ with respect to the instance given in Fig.~\ref{fig:spa-p-instance-1}. The resulting envy graph $G(M_2)$ is illustrated below.
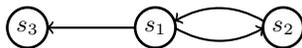
\begin{figure}[H]
\centering
\begin{tikzpicture}[scale=0.85]
\SetVertexNormal[MinSize = 1pt,LineWidth = 1pt,FillColor = white]
\Vertex[x=0, y=0]{$s_1$}[]
\Vertex[x=2, y=0]{$s_2$}
\Vertex[x=-2, y=0]{$s_3$} 
\draw [thick, style={->}] (-0.3, 0) to (-1.7, 0);
\draw [thick, style={->, bend right}](0.3,0) to (1.7,0); 
\draw [thick, style={->}, bend right] (1.7,0.1) to (0.3,0.1); 
\end{tikzpicture}
\label{fig:envygraph}
\caption[figure]{The envy graph $G(M_2)$ corresponding to the {\scriptsize SPA-P} instance in Fig.~\ref{fig:spa-p-instance-1}.}
\end{figure}
Clearly, $G(M')$ contains a directed cycle if and only if $M'$ admits a coalition. Moreover, $G(M')$ is acyclic if and only if it admits a topological ordering. Now to ensure that the matching $M$ obtained from a feasible solution $S$ under $\mathtt{J}$ is coalition-free, we will enforce $\mathtt{J}$ to encode the envy graph $G(M)$ and impose the condition that it must admit a topological ordering.
In what follows, we build on our IP model $\mathtt{J}$ of $\mathtt{I}$.
 
We create a binary variable $e_{i,i'}$ for each $(s_i, s_{i'}) \in \mathcal{S} \times \mathcal{S}$, $s_i \neq s_{i'}$, such that the $e_{i,i'}$ variables will correspond to the adjacency matrix of $G(M)$. For each $i$ and $i'$ ($1 \leq i \leq n_1, \; 1 \leq i' \leq n_1, \; i \neq i'$) and for each $j$ and $j'$ ($1 \leq j \leq n_2, \; 1 \leq j' \leq n_2$) such that $s_i$ prefers $p_{j'}$ to $p_j$, we impose the following constraint.
\begin{align}
\label{ineq:envy-variable}
e_{i,{i'}} + 1 \geq x_{i,j} + x_{{i'},{j'}\enspace.}
\end{align}

If $(s_i, p_j) \in M$ and $(s_{i'}, p_{j'}) \in M$ and $s_i$ prefers $p_{j'}$ to $p_j$, then $e_{i,{i'}} = 1$ and we say $s_i$ \textit{envies} $s_{i'}$. Otherwise, $e_{i,{i'}}$ is not constrained. Next we enforce the condition that $G(M)$ must have a topological ordering. To hold the label of each vertex in a topological ordering, we create an integer-valued variable $v_i$ corresponding to each student $s_i \in \mathcal{S}$ (and intuitively to each vertex in $G(M)$). We wish to enforce the constraint that if $e_{i,{i'}} = 1$ (that is, $(s_i, s_{i'}) \in A$), then $v_i < v_{i'}$ (that is, the label of vertex $s_i$ is smaller than the label of vertex $s_{i'}$). This is achieved by imposing the following constraint for all $i$ and $i'$ ($1 \leq i \leq n_1, \; 1 \leq i' \leq n_1, \; i \neq i'$).
\begin{align}
\label{ineq:topological-ordering}
\Aboxed{
v_i < v_{i'} + n_1(1 - e_{i,i'})\enspace.}
\end{align}
Note that the LHS of \eqref{ineq:topological-ordering} is strictly less than the RHS of \eqref{ineq:topological-ordering} if and only if $G(M)$ does not admit a directed cycle, and this implies that $M$ is coalition-free.
\paragraph{\textbf{Variables}.} We define a collective notation for each variable involved in $\mathtt{J}$ as follows.
\begin{center}
\begin{tabular}{lll}
$X = \{ x_{i,j}: 1 \leq i \leq n_1, 1 \leq j \leq n_2\}$, &  & $\Lambda = \{\alpha_{j}: 1 \leq j \leq n_2\}$, \\
$H = \{ \eta_{j,k}: 1 \leq j \leq n_2, 1 \leq k \leq n_3\}$, &  & $\Delta = \{\delta_{k}: 1 \leq k \leq n_3\}$, \\ 
$E = \{ e_{i,i'}: 1 \leq i \leq n_1, 1 \leq i' \leq n_1\}$, &  & $V = \{ v_{i}: 1 \leq i \leq n_1\}$\enspace.
\end{tabular} 
\end{center}
\paragraph{\textbf{Objective Function.}} The objective function given below is a summation of all the $x_{i,j}$ binary variables.  It seeks to maximize the number of students assigned (that is, the cardinality of the matching).
\begin{align}
\label{ineq:objective-function}
\Aboxed{\max \; \sum\limits_{i = 1}^{n_1} \; \sum\limits_{p_j \in A_i}x_{i,j}\enspace.}
\end{align}
Finally, we have constructed an IP model $\mathtt{J}$ of $\mathtt{I}$ comprising the set of integer-valued variables $X, \Lambda, H, \Delta, E \mbox{ and } V$, the set of constraints \eqref{ineq:studentassignment} - \eqref{ineq:topological-ordering} and an objective function \eqref{ineq:objective-function}. Note that $\mathtt{J}$ can then be used to solve {\scriptsize{MAX-SPA-P}} optimally.
Given an instance $\mathtt{I}$ of {\scriptsize{SPA-P}} formulated as an IP model $\mathtt{J}$ using the above transformation, we have the following lemmas.
\begin{lemma}
\label{lemma:solution-matching}
A feasible solution $S$ to $\mathtt{J}$ corresponds to a stable matching $M$ in $\mathtt{I}$, where $obj(S) = |M|$. 
\end{lemma}

\begin{proof}
Assume firstly that $\mathtt{J}$ has a feasible solution $S$. Let $M = \{(s_i, p_j) \in \mathcal{S} \times \mathcal{P}: x_{i,j} = 1\}$ be the assignment in $\mathtt{I}$ generated from $S$. Clearly $obj(S) = |M|$. We note that \eqref{ineq:studentassignment} ensures that each student is assigned in $M$ to at most one project. Moreover, \eqref{ineq:projectcapacity} and \eqref{ineq:lecturercapacity} ensures that the capacity of each project and lecturer is not exceeded in $M$. Thus $M$, is a matching. We will prove that \eqref{ineq:project-undersubscribed} - \eqref{ineq:type-c-blockingpair} guarantees that $M$ admits no blocking pair.

Suppose for a contradiction that there exists some acceptable pair $(s_i, p_j)$ that forms a blocking pair of $M$, where $l_k$ is the lecturer who offers $p_j$. This implies that $s_i$ is either unassigned in $M$ or prefers $p_j$ to $M(s_i)$. In either of these cases, $\sum_{p_{j'} \in S_{i,j}} x_{i,j'} = 0$, and thus $\theta_{i,j} = 1$. Moreover, as $(s_i, p_j)$ is a blocking pair of $M$, $p_j$ has to be undersubscribed in $M$, and thus $\sum_{i'=1}^{n_1} x_{i',j} < c_j$. This implies that the RHS of \eqref{ineq:project-undersubscribed} is strictly greater than $0$, and since $S$ is a feasible solution to $\mathtt{J}$, $\alpha_j = 1$.

Now suppose $(s_i, p_j)$ is a type 1(a) blocking pair of $M$. This implies $M(s_i) = p_{j''}$ for some $p_{j''} \in P_k$, where $l_k$ prefers $p_j$ tp $p_{j''}$. Thus $\gamma_{i,j,k} = \sum_{p_{j'} \in T_{k,j}} x_{i,j'} = 1$, which implies that the LHS of \eqref{ineq:type-a-blockingpair} is strictly greater than $2$. Thus $S$ is infeasible, a contradiction.

Next suppose $(s_i, p_j)$ is a type 1(b) blocking pair of $M$. This implies $s_i \notin M(l_k)$ and thus $1 - \beta_{i,k} = 1 - \sum_{p_j' \in P_k} x_{i,j'} = 1$. Also, $l_k$ has to be undersubscribed in $M$ which implies that the RHS of \eqref{ineq:lecturer-undersubscribed} is strictly greater than $0$, and thus $\delta_k = 1$. Hence the LHS of \eqref{ineq:type-b-blockingpair} is strictly greater than $3$, a contradiction, since $S$ is a feasible solution.

Next suppose $(s_i, p_j)$ is a type 1(c) blocking pair of $M$. This implies that $s_i \notin M(l_k)$ and thus $\beta_{i,k} = 0$. Also $l_k$ is full in $M$ and prefers $p_j$ to $p_z$, where $p_z$ is $l_k$'s worst non-empty project in $M$. This implies that the RHS of \eqref{ineq:lecturerfullconstraint} is strictly greater than $0$, and thus $\eta_{j,k} = 1$. Hence the LHS of \eqref{ineq:type-c-blockingpair} is strictly greater than $3$, and thus $S$ is infeasible, a contradiction.

Finally, we show that \eqref{ineq:envy-variable} and \eqref{ineq:topological-ordering} ensure that $M$ is coalition-free. Suppose for a contradiction that $M$ admits a coalition $\{s_{i_0}, \ldots, s_{i_{r-1}}\}$, for some $r \geq 2$. This implies that for each $t$ $(0 \leq t \leq r-1)$, $s_{i_t}$ prefers $M(s_{i_{t+1}})$ to $M(s_{i_t})$, where addition is taken modulo $r$, and hence $e_{i_t, i_{t+1}} = 1$, by \eqref{ineq:envy-variable}. It follows from \eqref{ineq:topological-ordering} that $v_{i_0} < v_{i_1} < \cdots < v_{i_{r-2}} < v_{i_{r-1}} < v_{i_r} = v_{i_0}$, a contradiction. Hence $M$ is coalition-free, and thus $M$ is a stable matching.
\end{proof}

\begin{lemma}
\label{lemma:matching-solution}
A stable matching $M$ in $\mathtt{I}$ corresponds to a feasible solution $S$ to $\mathtt{J}$, where $|M| = obj(S)$.
\end{lemma}

\begin{proof}
Let $M$ be a stable matching in $\mathtt{I}$. First we set all the binary variables involved in $\mathtt{J}$ to $0$. For all $(s_i, p_j) \in M$, we set $x_{i, j} = 1$. Now, since $M$ is a matching, it is clear that \eqref{ineq:studentassignment} - \eqref{ineq:lecturercapacity} is satisfied. For any acceptable pair $(s_i, p_j) \in (\mathcal{S} \times \mathcal{P}) \setminus M$ such that $s_i$ is unassigned in $M$ or prefers $p_j$ to $M(s_i)$, we set $\theta_{i,j} = 1$. For any project $p_j \in \mathcal{P}$ that is undersubscribed in $M$, we set $\alpha_j = 1$ and thus \eqref{ineq:project-undersubscribed} is satisfied. For \eqref{ineq:type-a-blockingpair} not to be satisfied, its LHS must be strictly greater than $2$. This would only happen if there exists $(s_i, p_j) \in (\mathcal{S} \times \mathcal{P}) \setminus M$, where $l_k$ is the lecturer who offers $p_j$, such that $\theta_{i,j} = 1$, $\alpha_j = 1$ and $\gamma_{i,j,k} = 1$. This implies that either $s_i$ is assigned in $M$ to a project $p_{j'}$ offered by $l_k$ such that $s_i$ prefers $p_j$ to $p_{j'}$, $p_j$ is undersubscribed in $M$, and $l_k$ prefers $p_j$ to $p_{j'}$. Thus $(s_i, p_j)$ is a type 1(a) blocking pair of $M$, a contradiction to the stability of $M$. Hence \eqref{ineq:type-a-blockingpair} is satisfied. 

Now for any lecturer $l_k \in \mathcal{L}$ that is undersubscribed in $M$, we set $\delta_k = 1$. Thus \eqref{ineq:lecturer-undersubscribed} is satisfied. Suppose \eqref{ineq:type-b-blockingpair} is not satisfied. This would only happen if there exists $(s_i, p_j) \in (\mathcal{S} \times \mathcal{P}) \setminus M$, where $l_k$ is the lecturer who offers $p_j$, such that $\theta_{i,j} = 1$, $\alpha_j = 1$, $\beta_{i,k} = 0$ and $\delta_{k} = 1$. This implies that either $s_i$ is unassigned in $M$ or prefers $p_j$ to $M(s_i)$, $s_i \notin M(l_k)$, and each of $p_j$ and $l_k$ is undersubscribed. Thus $(s_i, p_j)$ is a type 1(b) blocking pair of $M$, a contradiction to the stability of $M$. Hence \eqref{ineq:type-b-blockingpair} is satisfied.

Suppose $l_k$ is a lecturer in $\mathcal{L}$ and $p_j$ is any project on $l_k$'s preference list. Let $p_z$ be $l_k$'s worst non-empty project in $M$. If $l_k$ is full in $M$ and prefers $p_j$ to $p_z$, we set $\eta_{j,k} = 1$. Then \eqref{ineq:lecturerfullconstraint} is satisfied. Now suppose \eqref{ineq:type-c-blockingpair} is not satisfied. This would only happen if there exists $(s_i, p_j) \in (\mathcal{S} \times \mathcal{P}) \setminus M$, where $l_k$ is the lecturer who offers $p_j$, such that $\theta_{i,j} = 1, \alpha_j = 1, \beta_{i,k} = 0$ and $\eta_{j,k} = 1$. This implies that either $s_i$ is unassigned in $M$ or prefers $p_j$ to $M(s_i)$, $s_i \notin M(l_k)$, $p_j$ is undersubscribed and $l_k$ prefers $p_j$ to his worst non-empty project in $M$. Thus $(s_i, p_j)$ is a type 1(c) blocking pair of $M$, a contradiction to the stability of $M$. Hence \eqref{ineq:type-c-blockingpair} is satisfied.

We denote by $G(M) = (\mathcal{S}, A)$ the envy graph of $M$. Suppose $s_i$ and $s_{i'}$ are any two distinct students in $\mathcal
S$ such that $(s_i, p_j) \in M$, $(s_{i'}, p_{j'}) \in M$ and $s_i$ prefers $p_{j'}$ to $p_j$ (that is, $(s_i, s_{i'}) \in A$), we set $e_{i,i'} = 1$. Thus \eqref{ineq:envy-variable} is satisfied. Since $M$ is a stable matching, $M$ is coalition-free. This implies that $G(M)$ is acyclic and has a topological ordering $\sigma : \mathcal{S} \rightarrow \{1, 2, \ldots, n_1\}$. For each $i$ ($1 \leq i \leq n_1$), let $v_i = \sigma(s_i)$. Now suppose \eqref{ineq:topological-ordering} is not satisfied. This implies that there exist vertices $s_i$ and $s_{i'}$ in $G(M)$ such that $v_i \geq v_{i'} + n_1(1 - e_{i,i'})$. This is only possible if $ e_{i,i'} = 1$ since $1 \leq v_i \leq n_1$ and $1 \leq v_{i'} \leq n_1$. Hence $v_i \geq v_{i'}$, a contradiction to the fact that $\sigma$ is a topological ordering of $G(M)$ (since $(s_i, s_{i'}) \in A$ implies $v_i < v_{i'}$). Hence $S$, comprising the above assignment of values to the variables in $X \cup \Lambda \cup H \cup \Delta \cup E \cup V$, is a feasible solution to $\mathtt{J}$; and clearly $|M| = obj(S)$.
\end{proof}

Lemmas \ref{lemma:solution-matching} and \ref{lemma:matching-solution} immediately give rise to the following theorem regarding the correctness of $\mathtt{J}$.
\begin{restatable}[]{theorem}{correctness}
\label{thrm}
A feasible solution to $\mathtt{J}$ is optimal if and only if the corresponding stable matching in $\mathtt{I}$ is of maximum cardinality.
\end{restatable}
\begin{proof}
Let $S$ be an optimal solution to $\mathtt{J}$. Then by Lemma \ref{lemma:solution-matching}, $S$ corresponds to a stable matching $M$ in $\mathtt{I}$ such that $obj(S) = |M|$. Suppose $M$ is not of maximum cardinality. Then there exists a stable matching $M'$ in $\mathtt{I}$ such that $|M'| > |M|$. By Lemma \ref{lemma:matching-solution}, $M'$ corresponds to a feasible solution $S'$ to $\mathtt{J}$ such that $obj(S') = |M'| > |M| = obj(S)$. This is a contradiction, since $S$ is an optimal solution to $\mathtt{J}$. Hence $M$ is a maximum stable matching in $\mathtt{I}$. Similarly, if $M$ is a maximum stable matching in $\mathtt{I}$ then $M$ corresponds to an optimal solution $S$ to $\mathtt{J}$.
\end{proof}
    
\section{Empirical Analysis}
\label{sect:empirical-analysis}
In this section we present results from an empirical analysis that investigates how the sizes of the stable matchings produced by the approximation algorithms compares to the optimal solution obtained from our IP model, on {\scriptsize SPA-P} instances that are both randomly-generated and derived from real datasets. 
\subsection{Experimental Setup}
There are clearly several parameters that can be varied, such as the number of students, projects and lecturers; the length of the students' preference lists; as well as the total capacities of the projects and lecturers.
For each range of values for the first two parameters, we generated a set of random {\scriptsize SPA-P} instances. In each set, we record the average size of a stable matching obtained from running the approximation algorithms and the IP model. Further, we consider the average time taken for the IP model to find an optimal solution.

By design, the approximation algorithms were randomised with respect to the sequence in which students apply to projects, and the choice of students to reject when projects and/or lecturers become full. 
In the light of this, for each dataset, we also run the approximation algorithms 100 times and record the size of the largest stable matching obtained over these runs.  Our experiments therefore involve five algorithms: the optimal IP-based algorithm, the two approximation algorithms run once, and the two approximation algorithms run 100 times.

We performed our experiments on a machine with dual Intel Xeon CPU E5-2640 processors with 64GB of RAM, running Ubuntu 14.04. Each of the approximation algorithms was implemented in Java\footnote{\label{github}https://github.com/sofiat-olaosebikan/spa-p-isco-2018}. For our IP model, we carried out the implementation using the Gurobi optimisation solver in Java\textsuperscript{\ref{github}}. For correctness testing on these implementations, we designed a stability checker which verifies that the matching returned by the approximation algorithms and the IP model does not admit a blocking pair or a coalition. 
 
\subsection{Experimental Results}
\label{subsect:experimental-results}
\subsubsection{Randomly-generated Datasets.} 
All the {\scriptsize SPA-P} instances we randomly generated involved $n_1$ students ($n_1$ is henceforth referred to as the size of the instance), $0.5n_1$ projects, $0.2n_1$ lecturers and $1.1n_1$ total project capacity which was randomly distributed amongst the projects. The capacity for each lecturer $l_k$ was chosen randomly to lie between the highest capacity of the projects offered by $l_k$ and the sum of the capacities of the projects that $l_k$ offers.
In the first experiment, we present results obtained from comparing the performance of the IP model, with and without the coalition constraints in place.
\paragraph{\textbf{Experiment 0.}} We increased the number of students $n_1$ while maintaining a ratio of projects, lecturers, project capacities and lecturer capacities as described above. For various values of $n_1 \; (100 \leq n_1 \leq 1000)$ in increments of $100$, we created $100$ randomly-generated instances. Each student's preference list contained a minimum of $2$ and a maximum of $5$ projects. With respect to each value of $n_1$, we obtained the average time taken for the IP solver to output a solution, both with and without the coalition constraints being enforced. The results, displayed in Table~\ref{table:ip-time} show that when we removed the coalition constraints, the average time for the IP solver to output a solution is significantly faster than when we enforced the coalition constraints.

In the remaining experiments, we thus remove the constraints that enforce the absence of a coalition in the solution.  We are able to do this for the purposes of these experiments because the largest size of a stable matching is equal to the largest size of a matching that potentially admits a coalition but admits no blocking pair\footnote{This holds because the number of students assigned to each project and lecturer in the matching remains the same even after the students involved in such coalition permute their assigned projects.}, and we were primarily concerned with measuring stable matching cardinalities.  However the absence of the coalition constraints should be borne in mind when interpreting the IP solver runtime data in what follows.

In the next two experiments, we discuss results obtained from running the five algorithms on randomly-generated datasets.
	\paragraph{\textbf{Experiment 1.}} As in the previous experiment, we maintained the ratio of the number of students to projects, lecturers and total project capacity; as well as the length of the students' preference lists. For various values of $n_1 \; (100 \leq n_1 \leq 2500)$ in increments of $100$, we created $1000$ randomly-generated instances. With respect to each value of $n_1$, we obtained the average sizes of stable matchings constructed by the five algorithms run over the $1000$ instances. The result displayed in Fig.~\ref{experiment1} (and also in Fig.~\ref{experiment2}) shows the ratio of the average size of the stable matching produced by the approximation algorithms with respect to the maximum cardinality matching produced by the IP solver.

Figure~\ref{experiment1} shows that each of the approximation algorithms produces stable matchings with a much higher cardinality from multiple runs, compared to running them only once. Also, the average time taken for the IP solver to find a maximum cardinality matching increases as the size of the instance increases, with a running time of less than one second for instance size $100$, increasing roughly linearly to $13$ seconds for instance size $2500$ (see Fig.~\ref{experiment1time}).

\paragraph{\textbf{Experiment 2.}} 
In this experiment, we varied the length of each student's preference list while maintaining a fixed number of students, projects, lecturers and total project capacity. For various values of $x$ ($2 \leq x \leq 10$), we generated $1000$ instances, each involving $1000$ students, with each student's preference list containing exactly $x$ projects. The result for all values of $x$ is displayed in Fig.~\ref{experiment2}. Figure~\ref{experiment2} shows that as we increase the preference list length, the stable matchings produced by each of the approximation algorithms gets close to having maximum cardinality. It also shows that with a preference list length greater than $5$, the $\frac{3}{2}$-approximation algorithm produces an optimal solution, even on a single run. Moreover, the average time taken for the IP solver to find a maximum matching increases as the length of the students' preference lists increases, with a running time of two seconds when each student's preference list is of length $2$, increasing roughly linearly to $17$ seconds when each student's preference list is of length $10$ (see Fig.~\ref{experiment2time}).
\subsubsection{Real Datasets.}
The real datasets in this paper are based on actual student preference data and manufactured lecturer data from previous runs of student-project allocation processes at the School of Computing Science, University of Glasgow. Table \ref{table:realdatasets} shows the properties of the real datasets, where $n_1, n_2$ and $n_3$ denotes the number of students, projects and lecturers respectively; and $l$ denotes the length of each student's preference list. For all these datasets, each project has a capacity of $1$. In the next experiment, we discuss how the lecturer preferences were generated. We also discuss the results obtained from running the five algorithms on the corresponding {\scriptsize SPA-P} instances.
\paragraph{\textbf{Experiment 3.}} We derived the lecturer preference data from the real datasets as follows. For each lecturer $l_k$, and for each project $p_j$ offered by $l_k$, we obtained the number $a_j$ of students  that find $p_j$ acceptable. 
Next, we generated a strict preference list for $l_k$ by arranging $l_k$'s proposed projects in (i) a random manner, (ii) ascending order of $a_j$, and (iii) descending order of $a_j$, where (ii) and (iii) are taken over all projects that $l_k$ offers. Table~\ref{table:realdatasets} shows the size of stable matchings obtained from the five algorithms, where $A, B, C, D$ and $E$ denotes the solution obtained from the IP model, 100 runs of $\frac{3}{2}$-approximation algorithm, single run of $\frac{3}{2}$-approximation algorithm, 100 runs of $2$-approximation algorithm, and single run of $2$-approximation algorithm respectively. The results are essentially consistent with the findings in the previous experiments, that is, the $\frac{3}{2}$-approximation algorithm produces stable matchings whose sizes are close to optimal.
\begin{table}[H]
\caption{Results for Experiment 0. Average time (in seconds) for the IP solver to output a solution, both with and without the coalition constraints being enforced.}
\centering
\setlength{\tabcolsep}{0.2em}
\label{table:ip-time}
\begin{tabular}{c|c|c|c|c|c|c|c|c|c|c}
\hline
Size of instance & $100$ & $200$ & $300$ & $400$ & $500$ & $600$ & $700$ & $800$ & $900$  & $1000$\\ 
\hline
Av.\ time without coalition & $0.12$ & $0.27$ & $0.46$ & $0.69$ & $0.89$ & $1.17$ & $1.50$ & $1.86$ & $2.20$ & $2.61$\\ 
\hline 
Av.\ time with coalition & $0.71$ & $2.43$ & $4.84$ & $9.15$ & $13.15$ & $19.34$ & $28.36$ & $38.18$ & $48.48$    & $63.50$\\  
\end{tabular} 
\end{table}
\vspace{-0.6cm}
\begin{figure}[H]
\centering
\subfigure[]{%
\label{experiment1}%
\includegraphics[width=5.9cm]{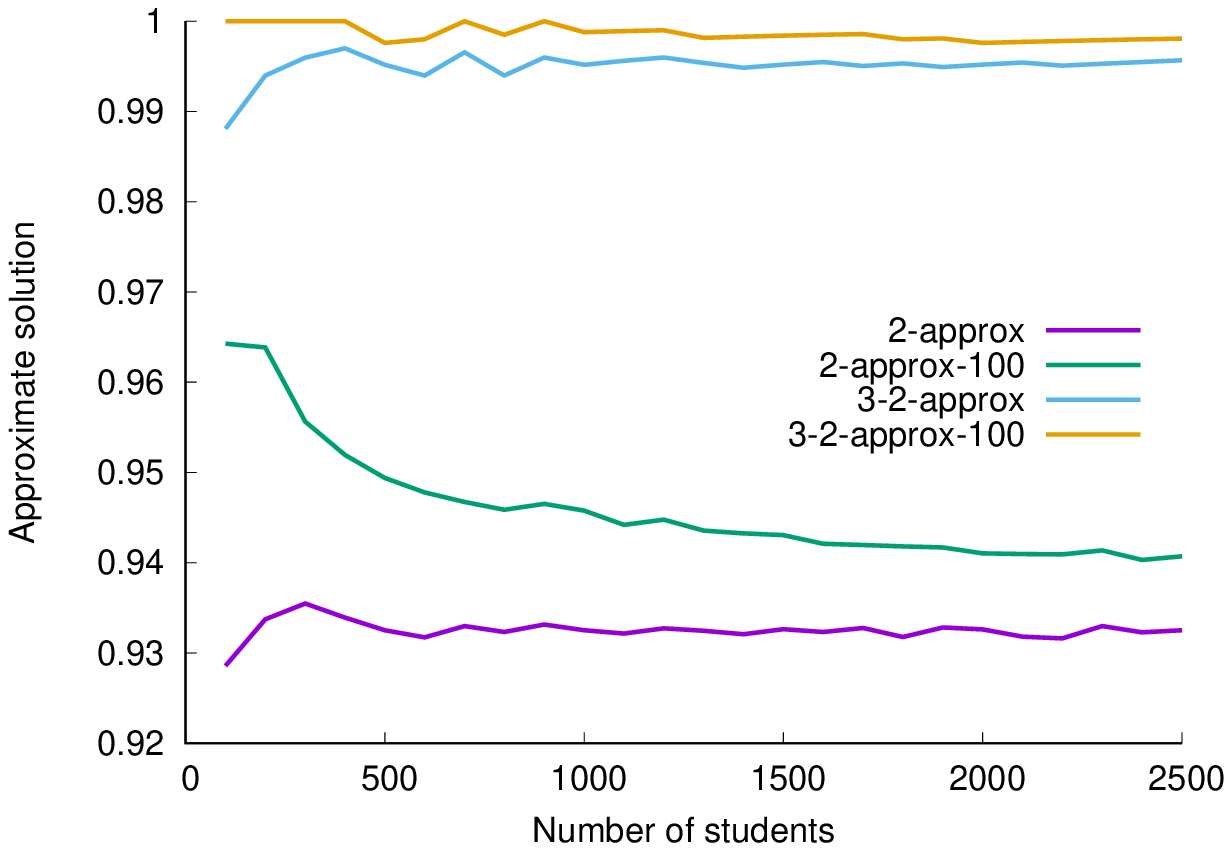}
}%
\quad
\subfigure[]{%
\label{experiment1time}%
\includegraphics[width=5.8cm]{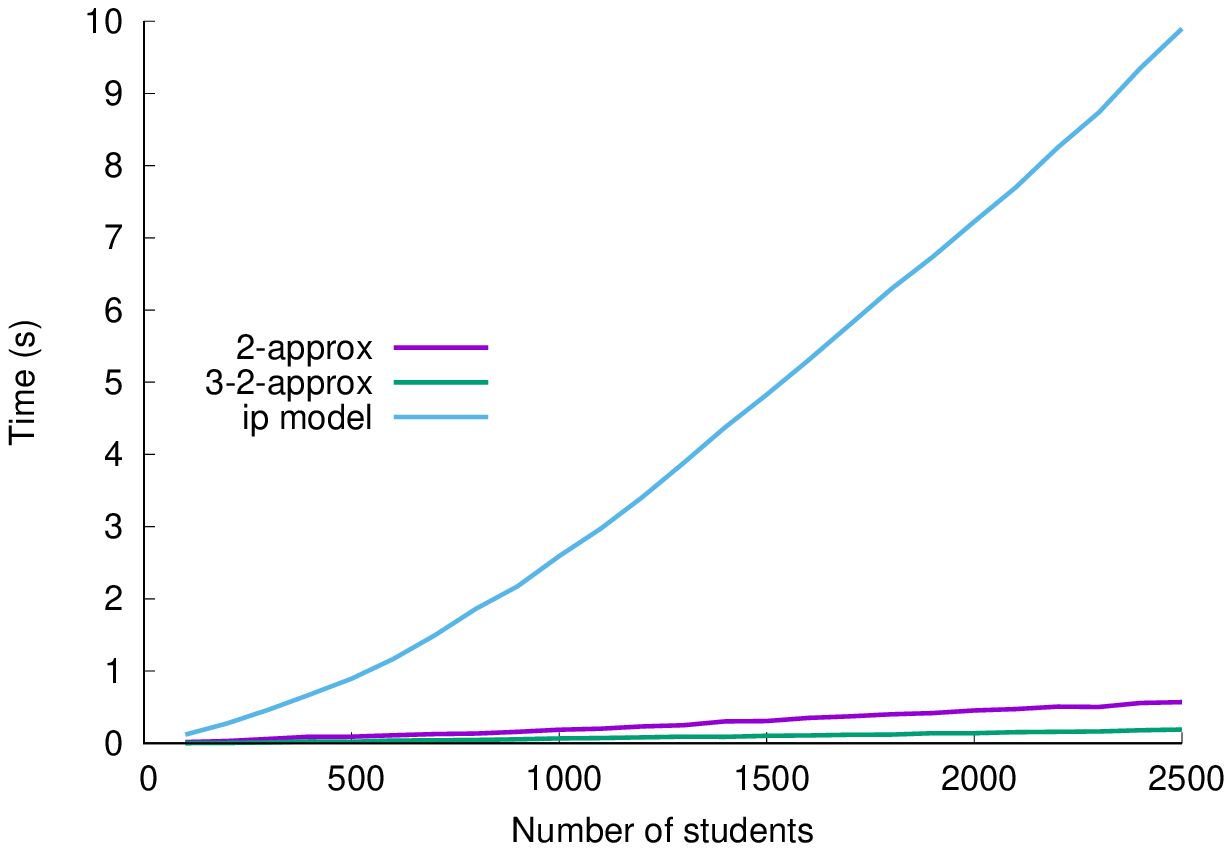}}%
\vspace{-0.2cm}
\caption{Result for Experiment 1.}
\label{fig1}
\end{figure}
\vspace{-0.9cm}
\begin{figure}[H]
\centering
\subfigure[]{%
\label{experiment2}%
\includegraphics[width=5.8cm]{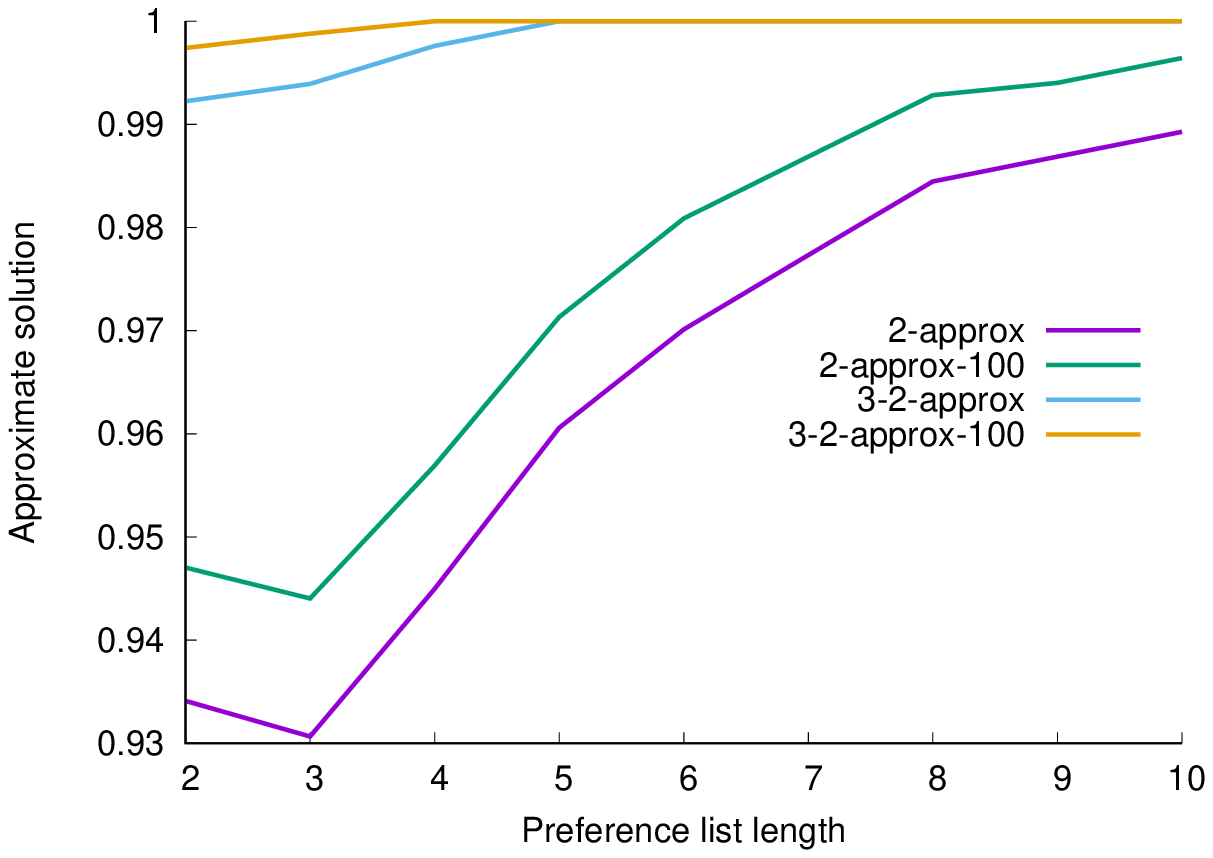}
}%
\quad
\subfigure[]{%
\label{experiment2time}%
\includegraphics[width=5.8cm]{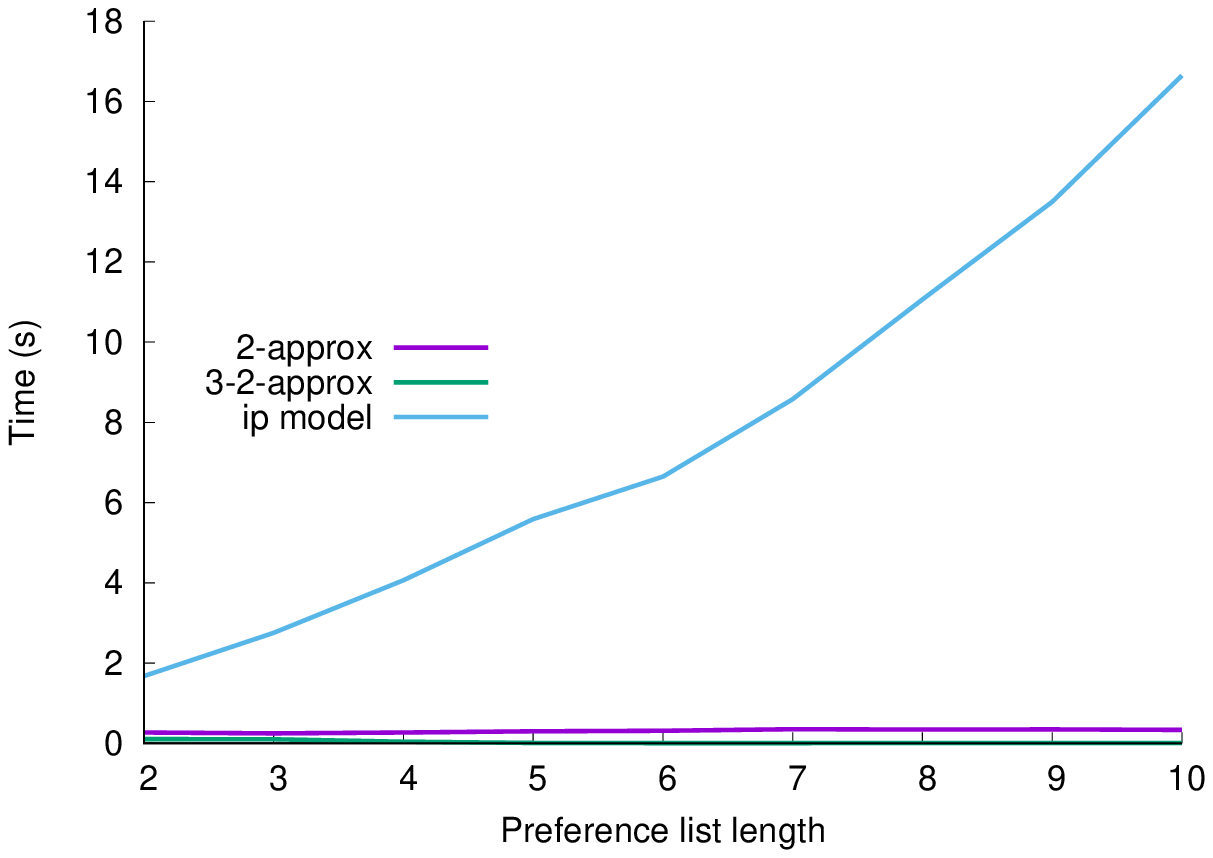}}%
\vspace{-0.2cm}
\caption{Result for Experiment 2.}
\label{fig2}
\end{figure}
\vspace{-0.6cm}
\begin{table}[H]
\caption{Properties of the real datasets and results for Experiment 3.}
\setlength{\tabcolsep}{0.4em}
\label{table:realdatasets}
\begin{tabular}{c|c|c|c|c||c|c|c|c|c||c|c|c|c|c||c|c|c|c|c}
\hline
\multicolumn{5}{c|}{}& \multicolumn{5}{c||}{Random} & \multicolumn{5}{c||}{Most popular} & \multicolumn{5}{c}{Least popular}\\
\hline
Year & $n_1$ & $n_2$ & $n_3$ & $l$ & $A$ & $B$ & $C$ & $D$ & $E$  & $A$ & $B$ & $C$ & $D$ & $E$  & $A$ & $B$ & $C$ & $D$ & $E$\\ 
\hline
\hline 
2014 & $55$ & $149$ & $38$ & $6$ & $55$ & $55$ & $55$ & $54$ & $53$     & $55$ & $55$ & $55$ & $54$ & $50$     & $55$ & $55$ & $55$ & $54$ & $52$\\ 
\hline 
2015 & $76$ & $197$ & $46$ & $6$ & $76$ & $76$ & $76$ & $76$ & $72$    & $76$ & $76$ & $76$ & $76$ & $72$    & $76$ & $76$ & $76$ & $76$ & $75$\\ 
\hline 
2016 & $92$ & $214$ & $44$ & $6$ & $84$ & $82$ & $83$ & $77$ & $75$    & $85$ & $85$ & $83$ & $79$ & $76$    & $82$ & $80$ & $77$ & $76$ & $74$\\ 
\hline 
2017 & $90$ & $289$ & $59$ & $4$ & $89$ & $87$ & $85$ & $80$ & $76$    & $90$ & $89$ & $86$ & $81$ & $79$    & $88$ & $85$ & $84$ & $80$ & $77$\\  
\end{tabular} 
\end{table}
\subsection{Discussions and Concluding Remarks}
The results presented in this section suggest that even as we increase the number of students, projects, lecturers, and the length of the students' preference lists, each of the approximation algorithms finds stable matchings that are close to having maximum cardinality, outperforming their approximation factor. Perhaps most interesting is the $\frac{3}{2}$-approximation algorithm, which finds stable matchings that are very close in size to optimal, even on a single run. These results also holds analogously for the instances derived from real datasets.

We remark that when we removed the coalition constraints, we were able to run the IP model on an instance size of $10000$, with the solver returning a maximum matching in an average time of $100$ seconds, over $100$ randomly-generated instances. This shows that the IP model (without enforcing the coalition constraints), can be run on {\scriptsize SPA-P} instances that appear in practice, to find maximum cardinality matchings that admit no blocking pair. Coalitions should then be eliminated in polynomial time by repeatedly constructing an \emph{envy graph}, similar to the one described in \cite[p.290]{Man13}, finding a directed cycle and letting the students in the cycle swap projects.

\bibliographystyle{plain}

\begin{thebibliography}{10}

\bibitem{AIM07}
D.J. Abraham, R.W. Irving, and D.F. Manlove.
\newblock Two algorithms for the {S}tudent-{P}roject allocation problem.
\newblock {\em Journal of Discrete Algorithms}, 5(1):79--91, 2007.

\bibitem{AB03}
A.A. Anwar and A.S. Bahaj.
\newblock Student project allocation using integer programming.
\newblock {\em IEEE Transactions on Education}, 46(3):359--367, 2003.

\bibitem{RGSA17}
R.~Calvo-Serrano, G.~Guill\'en-Gos\'{a}lbez, S.~Kohn, and A.~Masters.
\newblock Mathematical programming approach for optimally allocating students'
  projects to academics in large cohorts.
\newblock {\em Education for Chemical Engineers}, 20:11--21, 2017.

\bibitem{CFG17}
M.~Chiarandini, R.~Fagerberg, and S.~Gualandi.
\newblock Handling preferences in student-project allocation.
\newblock {\em Annals of Operations Research, \emph{to appear}}, 2018.

\bibitem{GS62}
D.~Gale and L.S. Shapley.
\newblock College admissions and the stability of marriage.
\newblock {\em American Mathematical Monthly}, 69:9--15, 1962.

\bibitem{HSVS05}
P.R. Harper, V.~de~Senna, I.T. Vieira, and A.K. Shahani.
\newblock A genetic algorithm for the project assignment problem.
\newblock {\em Computers and Operations Research}, 32:1255--1265, 2005.

\bibitem{IMY12}
K.~Iwama, S.~Miyazaki, and H.~Yanagisawa.
\newblock Improved approximation bounds for the student-project allocation
  problem with preferences over projects.
\newblock {\em Journal of Discrete Algorithms}, 13:59--66, 2012.

\bibitem{Kaz02}
D.~Kazakov.
\newblock Co-ordination of student-project allocation.
\newblock Manuscript, University of York, Department of Computer Science.
  Available from \url{http://www-users.cs.york.ac.uk/kazakov/papers/proj.pdf}
  (last accessed 8 March 2018), 2001.

\bibitem{Kir11}
Z.~Kir\'{a}ly.
\newblock Better and simpler approximation algorithms for the stable marriage
  problem.
\newblock {\em Algorithmica}, 60:3--20, 2011.

\bibitem{KIMS15}
A.~Kwanashie, R.W. Irving, D.F. Manlove, and C.T.S. Sng.
\newblock Profile-based optimal matchings in the {S}tudent--{P}roject
  {A}llocation problem.
\newblock In {\em Proceedings of IWOCA '14: the 25th International Workshop on
  Combinatorial Algorithms}, volume 8986 of {\em Lecture Notes in Computer
  Science}, pages 213--225. Springer, 2015.

\bibitem{Man13}
D.F. Manlove.
\newblock {\em Algorithmics of Matching Under Preferences}.
\newblock World Scientific, 2013.

\bibitem{MO08}
D.F. Manlove and G.~O'Malley.
\newblock Student project allocation with preferences over projects.
\newblock {\em Journal of Discrete Algorithms}, 6:553--560, 2008.

\bibitem{Rot84}
A.E. Roth.
\newblock The evolution of the labor market for medical interns and residents:
  a case study in game theory.
\newblock {\em Journal of Political Economy}, 92(6):991--1016, 1984.

\bibitem{ZZZ20}
\url{http://www.gurobi.com} ({G}urobi {O}ptimization website). Accessed
  09-01-2018.

\bibitem{ZZZ21}
\url{https://www.gnu.org/software/glpk} ({GNU} {L}inear {P}roramming {K}it).
  Accessed 09-01-2018.

\bibitem{ZZZ22}
\url{http://www-03.ibm.com/software/products/en/ibmilogcpleoptistud/} ({CPLEX}
  {O}ptimization {S}tudio). Accessed 19-05-2017.

\end{thebibliography}

\newcounter{mycounter}
\setcounter{mycounter}{\value{enumiv}}

\end{document}